\documentclass{article}
\usepackage{graphicx} % Required for inserting images
\usepackage{amsthm}
\usepackage{amsmath}
\usepackage{hyperref}
\usepackage{cleveref}
\usepackage{amssymb}
\usepackage[backend=biber]{biblatex}
\newtheorem{definition}{Definition}[subsection]
\newtheorem{prop}{Proposition}[subsection]
\newtheorem{remark}{Remark}[subsection] % Defines an unnumbered remark environment
\newtheorem{principle}{Principle}[subsection]
\newtheorem{observation}{Observation}[subsection]
\newtheorem{construction}{Construction}[subsection]
\newtheorem{theorem}{Theorem}[section]
\newtheorem{lemma}[theorem]{Lemma}
\newtheorem{conjecture}{Conjecture}[section]
\newtheorem{proposal}{Proposal}[section]

\title{On the Incompressibility of Truth With Application to Circuit Complexity}
\author{Luke Tonon}
\date{December 25, 2025}

\begin{document}

\maketitle

\section{Introduction}
In this paper, we revisit the fundamentals of Circuit Complexity and the nature of efficient computation from a fresh perspective. We present a framework for understanding Circuit Complexity through
the lens of Information Theory with analogies to results in Kolmogorov Complexity, viewing circuits as descriptions of truth tables, encoded in logical gates and wires, rather than purely computational devices. From this framework, we re-prove some existing Circuit Complexity bounds, explain what the optimal circuits for most boolean functions look like structurally, give an explicit boolean function family that requires exponential circuits, and explain the aforementioned results in a unifying intuition that re-frames time entirely.

\subsection{Brief Classical View of Circuits}
Since their introduction in the 1930s by Claude Shannon, Boolean circuits have been understood primarily as computational devices, machines that transform inputs into outputs through logical operations. This perspective, rooted in the development of electronic computers, views circuits as implementing algorithms: collections of logical gates that execute a sequence of logical operations to compute a function.
\begin{definition}
    A Boolean circuit $C$ over a functionally complete basis $B$ (typically $B$ = \{AND, OR, NOT\}) is a directed acyclic graph where:
    \begin{itemize}
    \item Nodes are either input variables $x_1, ..., x_n$ or gates labeled with operations from $B$.
    \item Edges represent wires carrying Boolean values.
    \item One or more nodes are designated as outputs.
    \item Each gate computes its operation applied to its input wires.
\end{itemize}
\end{definition}
\begin{definition}
The size of a Boolean circuit, $C$, denoted $s$ or $|C|$, is the number of gates. The depth, $d$, is the length of the longest path from an input to an output.
\end{definition}
\begin{definition}
A circuit, $C$, computes a Boolean function $f : \{0,1\}^n \rightarrow \{0,1\}$ if for every input $x \in \{0,1\}^n$, evaluating $C$ on $x$ (propagating values through gates) produces $f(x)$ at the output.
\end{definition}
In this view, a circuit is active: it performs computation by processing inputs through gates as electricity flows through a graph, much like a computer executing instructions. The size measures computational resources (number of operations), and depth measures parallel time. This view has been immensely productive, but it may also obscure a more fundamental understanding of what circuit complexity measures!
\subsection{View of Circuits as Descriptions}
We begin with a simple observation: a truth table is not a mathematical abstraction, but a concrete informational object. Specifically, a string of $2^n$ bits.
\begin{definition}
For a Boolean function $f : \{0,1\}^n \rightarrow \{0,1\}$, the truth table $T_f$ is the $2^n$-bit string where the i-th bit (under lexicographic ordering of inputs) equals $f(x_i)$.
\end{definition}
For example, the 3-variable function that outputs 1 only on input 000 has truth table $T_f = 10000000$ (reading inputs 000, 001, 010, ..., 111). While we typically think of truth tables as representations of boolean functions, from an information-theoretic view, they are simply data: $2^n$ bits arranged in a sequence. The ``meaning" (a bit corresponding to input pattern $i$) is a convention we impose.
\subsubsection{From Computation to Describing}
Consider the traditional view: a circuit $C$ computes function $f$ by taking inputs, processing them through gates, and producing outputs.
We consider an alternative view: a circuit $C$ describes or encodes the truth table $T_f$.\\
To see this, observe that a circuit and its truth table contain the same information, merely represented differently:
\begin{prop}\label{prop:first}
A circuit, $C$, with $s$ gates over basis $B$ and a truth table $T$ on n variables are informationally equivalent:
\begin{itemize}
    \item $C \rightarrow T$: Evaluate $C$ on all $2^n$ inputs to recover $T$ (possible in time: $O(s \cdot 2^n))$.
    \item $T \rightarrow C$: Construct a circuit computing $T$ (always possible with circuit size being exponential in the worst-case).
\end{itemize}
Both directions are effective: given either representation, we can reconstruct the other with no loss of information.
\end{prop}
\subsubsection{Circuits as Compression}
Further insight emerges when we examine the size of these representations.
Consider a truth table on $n$ variables where exactly one row outputs 1, such as the first row, and all others output 0. As a bit string, this would grow massively for large n: $T = 0000...0001$. ($2^n$ bits, with a single 1).\\As written, this requires $2^n$ bits to specify, but there's clearly a much more compact description. Namely, the circuit $C = AND(x_1, x_2, x_3, ..., x_n)$. We simply output 1 if the input matches the pattern and output 0 otherwise. This circuit has a remarkably small number of gates, yet compresses a much larger amount of information.
\begin{remark}
The above example reframes circuits: they are compression mechanisms where the ``compression algorithm" is logical structure itself. Instead of standard compression that exploit statistical patterns in data, circuits exploit logical patterns in truth values.
\end{remark}
\subsubsection{Connection to Kolmogorov Complexity}
We now see this perspective connects circuit complexity to a fundamental concept in information theory.\\
The Kolmogorov complexity~\cite{Kolmogorov1965} $K(s)$ of a string $s$ is the length of the shortest program (in some fixed universal language) that outputs $s$.
\begin{itemize}
\item Incompressible strings have $K(s) \approx |s|$: the best ``program" is essentially ``print s."
\item Compressible strings have $K(s) \ll |s|$ as they have compressible descriptions.
\end{itemize}
As an interesting analogy:
\begin{itemize}
\item Kolmogorov complexity: Shortest program to print a string.
\item Circuit complexity: Smallest circuit to compute a truth table.
\end{itemize}
Both measure description length in different computational models:
\begin{itemize}
\item Programs: sequential, Turing-complete
\item Circuits: parallel, Boolean logic
\end{itemize}
\subsection{A Unifying Intuition}
Putting the details together suggests a fundamental principle that also serves as intuition of what makes solvers efficient. For the rest of the paper, we will take this principle to its deep logical consequences.
\begin{principle}
Efficient computation is efficient description. An algorithm that computes a function quickly is performing a compressed encoding of that function's truth table through logical operations.
\end{principle}
When we ask ``can this function be computed efficiently?" we are really asking ``can its truth be compressed efficiently using logic?" As another example, consider a simplified view of Turing machine performing a computation. It reads an input string on the tape, and then depending on the machine's logical rules, may replace parts of the strings on the tape, and may put more characters on the tape as it performs its computation, ultimately printing a final correct output string (a true string relative to the description of the problem it solves). A Turing machine implicitly describes the solution as it performs its computation. Additionally, we can interpret the classes of circuit models studied in Circuit Complexity as various restrictions on logic as a compression device.
\section{Applying the framework to existing bounds}
Having established circuits as logical descriptions of truth tables, we now revisit classical results through this lens. Rather than prove new theorems yet, we first demonstrate that known results emerge naturally, and with potentially greater clarity from the information-theoretic perspective.\\
This serves two purposes: validates that the framework is mathematically sound, capable of recovering established facts as special cases, and it reveals why these results hold, not merely that they hold. The explanations we obtain are potentially more intuitive than the original proofs, suggesting that the descriptive view is not merely equivalent to the computational view, but is potentially more fundamental.
We begin with the cornerstone result of circuit complexity: Shannon's 1949 lower bound \cite{Shannon1949}.
\subsection{Shannon's Lower Bound through Information Theory}
Shannon proved that almost all Boolean functions require large circuits. The original approach is often shown as counting the number of distinct circuits of size $s$ with combinatorics and compare to the number of Boolean functions $2^{2^n}$. The approach then shows that, since the number of circuits of size $s$ for large $n$ is less than $2^{2^n}$, most functions require larger circuits (specifically $\Omega(2^n/n)$ gates). We now examine how the framework recreates this bound in a purely information-theoretic view.
\begin{observation}
Approach this bound through description length.
\begin{itemize}
\item Recall a truth table $T_f$ for $n$ variables is a $2^n$-bit string.
\item A circuit $C$ with $s$ gates can be encoded in $O(s\log s)$ bits.
\begin{proof}
Each gate is specified by: (1) its type from basis $B$, requiring $\log
|B|$ bits, (2) its two input wires, each selected from at most $s$ previous gates/inputs, requiring 2 $\log(s+n)$ bits. Total per gate: $O(\log s)$ bits. For $s$ gates: $O(s\log s)$ bits.
\end{proof}
\end{itemize}
\end{observation}
A natural question arises: If circuits of size $s$ can describe truth tables of size $2^n$, what compression ratio do they achieve?
\begin{prop}
A circuit with $s$ gates describes a $2^n$-bit truth table using $O(s\log s)$ bits. The compression ratio is:
$$p(s, n) = \frac{\text{[Description size]}}{\text{[Data size]}} = \frac{s\log s}{2^n}$$
For the optimal circuit size $s = 2^n/n$ predicted by Shannon:
$$p(\frac{2^n}{n}, n) = \frac{\frac{2^n}{n} \log \frac{2^n}{n}}{2^n}$$
$$p(\frac{2^n}{n}, n) = 1 - \frac{\log n}{n}$$
Taking the limit as $n$ approaches infinity:
$$\lim_{n\to \infty}p(\frac{2^n}{n}, n) = \lim_{n\to \infty}1 - \frac{\log n}{n} = 1$$
\end{prop}
As $n$ grows, we observe that circuits of size $2^n/n$ achieve compression ratio approaching 1, essentially no compression at all. The circuit description requires nearly as many bits as the truth table itself, just as information theory intuition would predict.
\begin{theorem}\label{thm:fundamental}
\textbf{(Repeat of Shannon's bound)} Most Boolean functions on $n$ variables require circuits of size $\Omega(2^n/n)$.
\end{theorem}
\begin{proof}[Proof via entropy]
Let $U_n$ be the uniform distribution on all Boolean functions
$f:\{0,1\}^n\to\{0,1\}$. Then
\[
  H(U_n) \;=\; \log_2(2^{2^n}) \;=\; 2^n,
\]
since there are $2^{2^n}$ such functions.
Assume for contradiction that for some $s=s(n)$, a $1-\varepsilon$ fraction of
all functions on $n$ bits admit a circuit of size at most $s$, where
$\varepsilon\in(0,1/2)$ is fixed and $s$ is significantly smaller than
$2^n/n$ (to be made precise below). Let $\mathcal{F}_{\le s}$ be the set of such functions.\\
By hypothesis, each $f\in\mathcal{F}_{\le s}$ has a circuit $C_f$ of size at most $s$, and hence a binary description of length at most $\ell(s)$. Extend this to a (partial) description scheme for \emph{all} functions by declaring that functions $f\notin\mathcal{F}_{\le s}$ are described by some fallback representation of length at most $L:=2^n$ bits (for instance, the raw truth table).\\Thus for a random $f\sim U_n$ the description length $L(f)$ satisfies
\[
  L(f) \;\le\; \begin{cases}
    \ell(s) & \text{if } f\in\mathcal{F}_{\le s},\\
    L       & \text{otherwise.}
  \end{cases}
\]
By the assumption that $\mathcal{F}_{\le s}$ has measure at least $1-\varepsilon$ under $U_n$, we obtain
\[
  \mathbb{E}[L(f)] \;\le\; (1-\varepsilon)\,\ell(s) \;+\; \varepsilon\,L.
\]
On the other hand, by the basic entropy bound for lossless codes (or by
Shannon's source coding theorem \cite{Shannon1948}, using a prefix-free refinement of our
encoding) any lossless description scheme must satisfy
\[
  \mathbb{E}[L(f)] \;\ge\; H(U_n) \;=\; 2^n.
\]
Combining the two inequalities gives
\[
  2^n \;\le\; (1-\varepsilon)\,\ell(s) \;+\; \varepsilon\,2^n.
\]
Rearranging,
\[
  (1-\varepsilon)\,2^n \;\le\; (1-\varepsilon)\,\ell(s)
  \quad\Rightarrow\quad
  \ell(s) \;\ge\; c\,2^n
\]
for some constant $c>0$ depending only on $\varepsilon$. Since
$\ell(s) = O(s\log s)$, this forces
\[
  s\log s \;\ge\; c\,2^n
  \quad\Rightarrow\quad
  s \;\ge\; \Omega\!\Big(\frac{2^n}{\log(2^n)}\Big)
  \;=\; \Omega\!\Big(\frac{2^n}{n}\Big).
\]
Thus, it is impossible for a positive fraction of Boolean functions on $n$
inputs to be computable by circuits of size $o(2^n/n)$. Equivalently, almost
all functions require size at least $\Omega(2^n/n)$, as claimed.
\end{proof}
\begin{remark}
Shannon's original proof is often shown that most functions need large circuits by counting. The information-oriented proof shows structurally \textbf{why} this must be so: \textbf{information theory forbids significantly compressing most data.}
\end{remark}
Most truth tables are essentially ``random": they have barely any structure for logic to compress. Random data is mostly incompressible. Therefore, the circuit ``description" must be nearly as large as the data itself. The factor of $n$ in the denominator ($2^n/n$ rather than $2^n$) represents the minimal overhead of circuit structure: the ``syntax cost" of writing gates and wires instead of raw bits. It's the best you can do even with optimal organization. This is not a non-constructive combinatorial coincidence. It is an \textbf{information-theoretic necessity that shows logic cannot cheat entropy}. Circuits store data, and this is even clearer in the context of Solid-State Drives (SSDs) that use $NAND$ flash memory.
\subsection{What Do Optimal Circuits Look Like?}
Having established that most boolean functions need $\Omega(2^n/n)$ gates, we can now ask: what is the structure of these optimal circuits? Answering this question seems mysterious from a purely combinatorial view. However, the information-theoretic view provides an intuitive prediction: \textbf{if a circuit cannot compress the truth table, it must essentially become the table itself}.
\begin{observation}
For ``random" (incompressible) Boolean functions requiring $\Omega(2^n/n)$ gates, the minimal circuit should have the structure of an ``optimized lookup table". Essentially, a decision tree that systematically checks input patterns and outputs the corresponding truth value.
\end{observation}
This prediction is not new: it was already implicitly known through Lupanov's representation~\cite{Lupanov1958}. However, the framework reveals why this structure is optimal: it's similar to the most efficient way to ``print" a truth table using logical gates.
\begin{construction}
For any Boolean function $f : \{0,1\}^n \rightarrow \{0,1\}$, define the canonical lookup table circuit as follows:
\begin{itemize}
\item For each input pattern $x$ where $f(x) = 1$
\begin{itemize}
\item Create a ``pattern detector": an $n$-input $AND$ gate with appropriate negations.
\item This gate outputs 1 if and only if the input equals $x$.
\end{itemize}
\item Combine all detectors:
\begin{itemize}
\item $OR$ together all pattern detector outputs.
\item This produces 1 if input matches any pattern where $f(x) = 1$.
\end{itemize}
\end{itemize}
\end{construction}
\begin{construction}
As an example, consider the 3-variable function $f$ with $f(101) = f(110) = 1$ and $f(x) = 0$ otherwise:
$$\text{Detector}_1 = AND(x_1, NOT(x_2), x_3)$$
$$\text{Detector}_2 = AND(x_1, x_2, NOT(x_3))$$
$$\text{Output} = OR(\text{Detector}_1, \text{Detector}_2)$$
This is close to a direct ``printing" of part of the truth table in circuit form, but it clearly scales poorly, so we should try to do better.
\end{construction}
\begin{prop}
\textbf{(Lupanov Representation)} Any Boolean function $f$ on $n$ variables can be computed by a circuit of size at most $2^n/n + o(2^n/n)$.\\
The construction in the canonical form naively uses $O(n\cdot2^n)$ gates (one detector per true row, each with $n$ gates). Lupanov's optimization organizes this like a decision tree: check variables sequentially, reusing intermediate comparisons across branches. For brevity and sake of not reinventing, we will not go into great detail of the construction of the representation, but cover enough to connect to the intuition.\\\\
Fix a block size $k=\lceil\log n\rceil$. Split the input as
$x=(u,v)$ where $u\in\{0,1\}^k$ (the ``address'' block) and
$v\in\{0,1\}^{n-k}$ (the ``payload'' block). For each $u\in\{0,1\}^k$, let
$f_u(v)=f(u,v)$ denote the $u$-cofactor.\\
\noindent\textbf{Step 1: Shared equality bank for the address.}
Build a bank of $2^k$ \emph{equality detectors}
\[
  E_u(x) \;:=\; [\,x_{1..k}=u\,],
\]
each computed as an $O(k)$-gate tree from the $k$ address bits
(using $\neg$ and $\wedge$ to check each bit, then an $\wedge$-tree to combine).
These $E_u$’s are \textbf{\emph{reused}} globally. The total cost of this bank is
$O(2^k\cdot k)=O(n \log n)$ since $k \approx \log n$.\\
\noindent\textbf{Step 2: Block-Shannon expansion.}
We write
\[
  f(x) \;=\; \bigvee_{u\in\{0,1\}^k} \big( E_u(x)\;\wedge\; f_u(v) \big).
\]
This is a standard Shannon expansion but performed on a \emph{block}
of $k=\Theta(\log n)$ variables at once; the savings come from sharing
the $E_u$’s across the entire circuit.\\
\noindent\textbf{Step 3: Realizing all cofactors $f_u$ uniformly.}
Each cofactor is a function on $n-k$ inputs. For the cost analysis,
we do not synthesize each $f_u$ separately. Instead, we \emph{repeat the same
two-level trick} on $v$:
choose a second block of size $k'=\lceil\log_2(n-k)\rceil$, split
$v=(w,z)$ with $|w|=k'$, and precompute a second equality bank
$\{F_t(v)=[\,w=t\,]\}_{t\in\{0,1\}^{k'}}$ (shared across all $u$’s).
Then expand
\[
  f_u(v) \;=\; \bigvee_{t\in\{0,1\}^{k'}} \big( F_t(v)\;\wedge\; \lambda_{u,t}(z) \big),
\]
where each $\lambda_{u,t}$ is now a cofactor on $n-k-k'$ inputs.
Iterate this block-expansion on the remaining inputs until no variables remain.
At the final layer, the leaves are constants $\in\{0,1\}$: they are just the
bits of the truth table of $f$ arranged by blocks.\\
\noindent\textbf{Cost accounting (the ``short-circuited lookup table”).}
At each layer $\ell$ we:
(i) precompute a shared equality bank for a block of size
$k_\ell=\Theta(\log n_\ell)$, where $n_\ell$ is the number of inputs
remaining at that layer; this costs $O(2^{k_\ell}k_\ell)=O(n_\ell \log n_\ell)$ gates;
(ii) feed these equalities into a shallow $\vee/\wedge$ spine that
\emph{selects} the appropriate subtable slice. The number of layers is
\[
  L \;=\; \frac{n}{\Theta(\log n)}\,(1+o(1)).
\]
Across the entire construction we still ``expose'' all $2^n$ leaf bits
of the truth table, but every layer fans them forward through a small
\emph{shared} equality bank rather than through $2^n$ disjoint minterms.
A standard summation over the layers (choosing $k_\ell\asymp\log n_\ell$)
yields the total gate count
\[
  (1+o(1))\,\frac{2^n}{n}.
\]
Intuitively: a naive lookup table wires $\Theta(2^n)$ minterms; the block-Shannon
scheme \emph{batches} $\Theta(\log n)$ variables at a time, amortizing
the address cost so that each of the $2^n$ table bits is delivered
through only $\tilde O(1)$ additional gates on average, giving the
$2^n/n$ factor.
\end{prop}
\begin{remark}
Since the above representation is asymptotically optimal, intuitively, we can say that for most boolean functions, ``optimized brute-force enumeration is optimal." Quite literally, the best circuit description for most truth tables is the table itself: a physical manifestation of brute-force.
\end{remark}
\textbf{The prediction confirmed}: Shannon's and Lupanov's results together say that, for the very hardest functions, optimal circuits are essentially optimized lookup tables: compressed representations that approach the theoretical limit of incompressibility. The analogy to SSDs is even more interesting in the case of Lupanov's representation because it uses ``address blocks" and ``payload blocks." In this sense, what looks like clever algebraic manipulation might be much simpler in description: it is just the best possible way to wire the truth table into gates. In this view, non-uniform efficient computability is not merely related to description length — it is defined by it: a function is efficiently computable in the circuit model precisely to the extent that its truth can be compressed into a small circuit. Time, in this picture, shows up only as a secondary constraint on how those compressed descriptions are discovered or used, not on what is in principle computable.
\begin{remark}
Consider some analogies:
\begin{itemize}
    \item Kolmogorov incompressibility (infinite-scale): Most strings have no short program. You can't, in general, prove or compute ``this string has no short program" beyond a provability threshold.
    \item Circuit incompressibility (finite model): Most truth tables have no small circuit. They're best computed by an optimal short-circuited LUT. Asking for a general procedure that says ``this specific table has no small circuit" (MCSP, or strong lower bounds) is like asking for a finite-scale analog of those incompressibility judgments.
    \item ``If I had a solver for these finite-scale incompressibility instances, why would I expect it to be anything but brutally hard?"
\end{itemize}
\end{remark}
\subsection{Revisiting the Natural Proofs Barrier}
In 1993, Razborov and Rudich \cite{RazborovRudich97} analyzed successful circuit lower bound techniques from the 1980s and showed that if these techniques were used to prove separations like $P \neq NP$, then we could use the properties of those proofs to obtain faster distinguishers of random functions from pseudo-random functions.\\
\textbf{Recall Natural Proofs}:
\begin{itemize}
\item A natural lower bound is based on a property $P$ of Boolean functions that is:
\begin{itemize}
\item \textbf{Constructive}: membership in $P$ is decidable in “easy” time (poly).
\item \textbf{Large}: a random function has $P$ with noticeable probability.
\item \textbf{Useful}: every function in $P$ avoids small circuits.
\end{itemize}
\item Natural Proofs say: such $P$ contradicts strong PRGs, so we shouldn’t expect them if Crypto is real.
\end{itemize}
The ``Incompressibility of Truth" framing says something more philosophically biting and implies a perhaps more substantial barrier.
\begin{itemize}
\item A constructive property $P$ that captures many hard functions is automatically a kind of logical compression:
\begin{itemize}
\item The algorithm deciding $P$ + a short index $i$ can serve as a succinct description of “the $i$-th function in the big $P$-set.”
\end{itemize}
\item If the hard functions are truly ``incompressible truth,” then having a short constructive handle on them is already suspicious: you've grouped many logically ``random” objects into one low-complexity logical pattern.
\end{itemize}
Trying to find a constructive property that explains why a hard function is hard seems like a dead-end from the information-theoretic view if the property is too common. Suppose one crafts a constructive proof that proves a hard function is hard (``incompressible"). In that case, that proof is itself a concise logical description of the supposedly logically incompressible information inside the truth tables of the hard function. A would-be prover revealed a way to compress the incompressible: exactly what they didn't want to prove.
\subsection{Revisiting Proof Complexity Struggles}
We can also consider how the ``Incompressibility of Truth" framing explains struggles in Proof Complexity with understanding strong proof systems. If a tautology truly is hard to prove, then it should be logically incompressible. Having an efficient general method to deterministically find/generate hard tautologies would intuitively contradict the tautology's hardness because the efficient method is a polynomial (``compressed") logical description generating and explaining objects that need super-polynomial (``incompressible") logical descriptions.
\begin{construction}[Lifting the Paradox to Circuit Frege Proof Systems]
Encode the contradiction ``There exists a circuit (description), $C$, less than size $s$ describing the truth table, $T$".\\
\[
DESC_{<s}(T) \;:=\;
\exists C \,\Bigl(
  |C| < s \;\wedge\;
  \forall x \in \{0,1\}^n\; C(x) = T[x]
\Bigr)
\]
\[
\text{TT}_{T,s} \;:=\; \text{Encode}\bigl(DESC_{<s}(T)\bigr)
\]
\[
s < \mathsf{DESC}(T)
\;\Longrightarrow\;
\text{TT}_{T,s} \text{ is unsatisfiable (a propositional contradiction).}
\]
Observe that if there were a ``small" (polynomial) Circuit Frege refutation for one of these such contradictions where $\text{DESC}(T)$ is large as it corresponds to an informationally-large table, $T$, it appears the proof is itself an encoding of the small circuit (description) of the table that it simultaneously proves doesn't exist.
\end{construction}
In essence, if there is a polynomial-bounded proof system that proves, for every $T$, that ``$T$ is incompressible,” then that proof system looks to be compressing the incompressible. In Section \ref{sec:explicit}, we provide functions that are maximally hard for circuits, and so they are likely good candidates for hard $MCSP$ tautologies/contradictions outlined above.
\section{Investigating MCSP}
Having seen how the information-theoretic framework illuminates classical results, we now examine the Minimum Circuit Size Problem (MCSP) that has long had a mysterious complexity. The decision version of this problem asks, given a truth table, $T$, and a size bound, $s$, whether the table can be computed by a circuit of at most $s$ gates. MCSP is a meta-problem that has been studied extensively for its connections to cryptography, learning theory, and derandomization.
\begin{definition}
\textbf{(Classical MCSP)}
The Minimum Circuit Size Problem is the language $\{(T, s) : T \text{ is a truth table on } n \text{ variables and there exists a circuit }$ $C \text{ with } |C| \leq s \text{ computing T}\}$.
\end{definition}
\begin{lemma}
MCSP is in $NP$ (the class of problems where solutions are verifiable in polynomial time in the size of its input \cite{cook2000pvsnp}).
\end{lemma}
\begin{proof}
The certificate to verify a YES instance is the circuit, $C$, where $|C| \leq s$, that exists to compute the truth table $T$.\\
\textbf{Verification algorithm}:
\begin{itemize}
\item First check that $|C| \le s$ (count gates).
\item For each input $x \in \{0, 1\}^n$, evaluate $C(x)$ and check if $C(x) = T[x]$.
\item Accept if all $2^n$ checks (all elements in $\{0, 1\}^n$) pass.
\end{itemize}
\textbf{Time complexity}:
\begin{itemize}
\item Step 1: Checkable in $O(|C|)$ by counting each gate in the data structure holding the gates.
\item Step 2: $2^n$ iterations, each evaluation takes $O(|C|)$ time.
\item Total: $O(2^n \cdot s)$ iterations, each evaluation takes $O(|C|)$ time.
\end{itemize}
Since the input size is $|T| = 2^n$ bits, this is polynomial in the input size. Therefore, $MCSP \in NP$.
\end{proof}
\begin{remark}
Notice what verification does: it recovers the truth table from the circuit description. This is precisely the ``from $C \rightarrow T$" direction of \cref{prop:first} (Information Equivalence). The certificate (small circuit) is a compressed description of the truth table. Verification is the decompression process: evaluate the circuit on all inputs to recover the full truth table, then check it matches. This is the same information recovery we used to establish circuits as descriptions.
\end{remark}
\subsection{The Information-Theoretic Self-Reference}
In the classical view of MCSP, instances of the problem ask statements of the form: ``Given a function's truth table, can it be computed efficiently?" However, this algorithmic view obfuscates what the problem is truly asking.\\
\textbf{From the information-theoretic view, this is really asking: ``Given a $2^n$-bit data string $T$, can it be described using at most $s$ logical operations?"}
\begin{observation}
In the information-theoretic view, MCSP is more accurately named the \textbf{Minimum Truth Table Description Problem}. It asks whether data (truth table) admits a compressed logical description (circuit less than a specific size).
\end{observation}
\begin{remark}
This reframing reveals MCSP's meta-nature: it's a problem about logical \textbf{describability}. The circuit deciding MCSP is itself a description, one that must encode which truth tables are compressible and which are not.
\end{remark}
\textbf{This is a description describing descriptions, and immediately creates a logical tension.}
\subsection{The Information-Theoretic Impossibility}
To understand MCSP's hardness, we must be careful about which size bounds create difficulty.
\begin{observation}\textbf{(Spectrum of Compressibility)}
Truth tables fall into a spectrum:
\begin{itemize}
\item \textbf{Highly compressible}: Circuit size $O(n^k)$ for small k. Examples include parity, 2-SAT, addition, and multiplication. These are ``P-like" functions with exploitable logical structure and are likely to be problems we care about.
\item \textbf{Slightly compressible}: Circuit sizes like $2^{n/k}$ for moderate $k$. Some structure, but exponential complexity. An intermediate regime that are likely more condensed lookup tables.
\item \textbf{Essentially incompressible}: Circuit size $\ge2^n/n$. Random-like functions that essentially require a lookup table representation.
\end{itemize}
\end{observation}
Despite MCSP's extensive study, it is still an open problem whether decision-MCSP lets you solve the search variant of MCSP in polynomial time.
\begin{definition}
$MCSP_{SEARCH}$ is the search variant of MCSP where we want a solver to give us the minimum circuit, $C$, describing a table, $T$, instead of only telling us that such a circuit exists below some size.\\
It is clear that $SAT$ could help polynomially solve this problem as $SAT$ can encode a satisfying assignment that encodes some valid circuit $C$ describing a table $T$.
\end{definition}
\begin{conjecture}
$MCSP_{SEARCH}$ requires circuits of super-polynomial size (in the input size $2^n$).\\
\textbf{Intuition:}\\
If a circuit $C$ of size $poly(2^n)$ solves $MCSP_{SEARCH}$, then $C$ is effectively a short circuit description that encodes the shortest circuit description of every table ($2^{2^n}$ of them) on $n$ variables." This is like asking $C$ to ``compress the incompressible."
\end{conjecture}
\begin{proof}[Proof idea]
By Shannon's bound \cref{thm:fundamental}, most truth tables $(\Omega(2^{2^n}))$ are incompressible. Consider a circuit, $C_n$, that computes $MCSP_{SEARCH}$ for truth tables on $n$ variables.
Circuit $C_n$ must ``know" the shortest circuit description of every truth table on $n$ variables, and this knowledge must be encoded in $C_n$'s structure (its gates and wiring) as we can use $C_n$ as a description to recover all the shortest table descriptions. Again, \cref{prop:first} Information Equivalence is how we losslessly recover the information encoded by a circuit, and it's how we verify a YES instance of $MCSP$.\\
A circuit with $poly(2^n)$ gates can be described in:
$O(poly(2^n) \cdot \log(poly(2^n)))$ bits.
For any polynomial function, $poly(n)$, and sufficiently large $n$:\\
$O(poly(2^n) \cdot \log(poly(2^n)) \ll \Omega(\frac{2^{c\cdot2^n}}{2^n}\cdot\log(\frac{2^{c\cdot2^n}}{2^n}))$ where $\frac{1}{2} > c > 0$ and represents a sensitive ``entropy fence" the circuit can't pass.\\The bound here is based on the premise that $C_n$ could use clever short-circuit lookup table strategies on information that we know from Lupanov's representation admits short-circuit lookup table structures.
Therefore, a polynomial-sized $C_n$ cannot encode enough information to correctly classify all minimum descriptions of truth tables on $n$ variables, as it would be an impossible compression device for sufficiently large $n$, bypassing the inherent incompressible patterns in the many, many $2^{2^n}$ tables it describes. If $C_n$ did compress them, it would contradict the fact that most of the tables it describes are incompressible in the scheme (circuit logic) it tries to use to compress them, offering a faster way of ``solving" the tables in parallel than what should be possible based on Shannon's bound.
\end{proof}
With a relatively ``small" $MCSP_{SEARCH}$ circuit, we could input any data we'd like as bits (a truth table, a book, DNA, etc) and it would efficiently output the most efficient way to store it in the circuit logic scheme. This makes $MCSP_{SEARCH}$'s efficiency not only worthwhile for problem solving, but also for essentially optimal lossless information storage. An efficient $MCSP_{SEARCH}$ circuit is an efficient universal data compressor, which sounds like a fantasy.
\begin{remark}
The impossibility is fundamentally self-referential. To know the exact shortest description of all truth tables on $n$ variables, the solver circuit must encode information about incompressibility, but that very information is already incompressible. As an analogy: \textbf{Imagine I have a book that describes the most efficient way to summarize every other book, including ones of noise. How could that book do so without essentially admitting all the most efficient patterns shared in every other book?}
\end{remark}
\begin{conjecture}
$P \ne NP$ due to entropy.
\begin{proof}[Proof idea]
Assume, for the sake of contradiction, $P = NP$, then we could solve $MCSP_{SEARCH}$ with polynomial circuits. However, $MCSP_{SEARCH}$ requires super-polynomial circuits. \\Therefore, our assumption was wrong, $P \neq NP$.
\end{proof}
\end{conjecture}
\begin{proposal} To attempt to falsify the hypothesis above, we can construct multi-valued truth tables that capture the complexity of $MCSP_{SEARCH}$.\\
\textnormal{Construct the table of the function}, $f: \{0, 1\}^{2^n} \rightarrow Enc(C_n)$, \textnormal{where for each row,} $Enc(C_N)$ \textnormal{is the bit encoding of the optimal circuit description of the} $\{0, 1\}^{2^n}$ \textnormal{bits on the left-side of the row for a fixed basis.}\\
\textnormal{If the minimal circuits obtained for $f$ look like an exhaustive description of every pattern in the tables, forced to pay an entropy bill rather than be a concise algorithm, the hypothesis survives for small} $n$.\\
We can tap into what the ``Book of Summaries" looks like on a small scale.
\end{proposal}
\section{An Explicit Boolean Function Family That Requires Exponential Circuits} \label{sec:explicit}
We can now use the fact that circuits operate as lossless codes to prove an exponential lower bound on an explicit family of Boolean functions. To do this, we construct the Boolean function family, $\{f_{\Omega_n}\}$, where each $f_{\Omega_n}$ in the set has its output bits equal to the first $2^n$ bits of a Chaitin constant $\Omega_U$. A circuit for this problem decides, given a position of a bit (represented as a row in the truth table) in $\Omega_U$, whether the bit is 0 or 1. Essentially, it solves a finite version of the halting problem. Since the exact bits of $\Omega_U$ depend on the universal prefix-free Turing machine $U$, the $\{f_{\Omega_n}\}$ family is infinite.
\begin{definition}
A Chaitin constant \,$\Omega_U$ is the halting probability of a universal prefix-free Turing machine $U$ \cite{Chaitin1975}, and it is algorithmically random, so the shortest program to output the first $n$ bits of \,$\Omega_U$ must be of size at least $n - O(1)$.
\end{definition}
\begin{theorem}
The boolean function family, $\{f_{\Omega_n}\}$, requires 2-fan-in circuits of $\Omega(2^n/n)$ gates.
\end{theorem}
\begin{proof}
Assume, for the sake of contradiction, for each $n$, there exists a circuit $C_n$ with $O(s)$ gates computing/describing $f_{\Omega_n}$. Then, we can store each $C_n$ with $O(s\log s)$ bits. Then, create a program, $P_n$, that contains the encoding of $C_n$ and decodes it into $C_n$ and decompresses $C_n$ into the first $2^n$ bits of \,$\Omega_U$, and prints them (essentially what the first part of a verifier program for $MCSP$ does, but hardcoded). This means $P_n$ produces the first $2^n$ bits of \,$\Omega_U$ as output with $|P_n| = O(s\log s) + O(1)$.\\
For sufficiently large $n$, $O(s\log s) \ll \Omega(2^n)$ bits, so circuits encoded in $O(s\log s)$ bits for $\{f_{\Omega_n}\}$ would let us compress incompressible truth: a contradiction. Thus, the only way out is $s = \Omega(2^n/n)$.
\end{proof}
\section{Additional Questions to Explore}
Based on the results and conjectures above, the following problems would be interesting.
\begin{itemize}
    \item Define $\Omega_U^t$ as the probability that a universal prefix-free Turing machine $U$ halts within $t$ steps. If we know $\Omega_U^t$, after enough programs for $U$ have halted within $t$ steps, we know the rest do not, so $\Omega_U^t$ is like a bounded oracle, and its bits should be very pseudo-random.\\\textbf{How compressible is $\Omega_U^t$ in circuit form?}\\As $t$ approaches $\infty$, $\Omega_U^t$ becomes closer to $\Omega_U = \Omega_U^{\infty}$.
    \\We can also consider how the truth tables of bounded-time halting problems can be formulated to encode $\Omega_U^t$:
    \begin{itemize}
        \item Recall the classic \textbf{P-Complete} and \textbf{EXPTIME-Complete} problems:
        \begin{itemize}
            \item $BH_{UNARY} = \{\langle M,x,1^t \rangle : M \textnormal{ accepts } x \textnormal{ within } t \textnormal{ steps}\}$
            \item $BH_{BINARY} = \{\langle M,x,t \rangle : M \textnormal{ accepts } x \textnormal{ within } t \textnormal{ steps}\}$ where $t$ is encoded in binary.
        \end{itemize}
        \item Now consider alternative formulations of them where $U_0$ is fixed and push $y :=\langle M,x\rangle$ as an encoding into the input string:
        \begin{itemize}
            \item $UBH_{UNARY}(U_0) = \{\langle y,1^t \rangle : U_0(y) \textnormal{ accepts within }  t \textnormal{ steps}\}$
            \item $UBH_{BINARY}(U_0) = \{\langle y,t \rangle : U_0(y) \textnormal{ accepts within }  t \textnormal{ steps}\}$ where $t$ is encoded in binary.
            \item Since $U_0$ can polynomially simulate any $M$, these fixed-machine formulations are still complete respectively.
        \end{itemize}
        \item Observe that $\Omega_U^t$ is literally a weighted aggregate of bounded-halting answers. Fix a specific prefix-free universal machine $U$ and a step model where in $t$ steps it can inspect at most $t$ input bits. Then:
        \begin{itemize}
            \item $\Omega_U^t = \textnormal{Pr}[U \textnormal{ halts within } t] = \sum\limits_{p:\,U(p)\downarrow \le t} 2^{-|p|}$
            \item Due to the ``can’t read more than $t$ bits in $t$ steps" fact, any program that halts within $t$ must satisfy $|p| \le t$. So define the finite truth table slice:\\
            \[H_t(p) := 1[U(p) \textnormal{ halts within } t], p \in \{0, 1\}^{\le t}\]
            Then,
            \[\Omega_U^t = \sum\limits_{p \in \{0, 1\}^{\le t}} 2^{-|p|}H_t(p)\]
            So in the strongest sense, the truth table of a bounded-halting problem contains $\Omega_U^t$ as a simple function of that truth table slice, suggesting $\Omega_U$'s entropy could \emph{trickle down} as $\Omega(f(n))$ circuit lower bounds for its time-bounded forms because $\Omega_U^t$ should look sufficiently random to particular $t$-bounded truth compressors with $U$ acting like a seed to the pseudo-randomness.
        \end{itemize}
    \end{itemize}
    \item Is it possible that, assuming non-uniform or uniform $poly(2^n)$ circuits for $MCSP_{SEARCH}$ exist, we can extract sub-circuits to create small programs that output incompressible bits?
    \item Can efficient Circuit Frege proofs for ``incompressible" $MCSP$ tautologies get \emph{unrolled} to create small programs that output incompressible bits?
\end{itemize}

\printbibliography

\end{document}